 \documentclass[onecolumn,11pt]{IEEEtran}
\usepackage[a4paper, top=1.2in, left=1.2in, right=1.2in]{geometry}
\usepackage{graphicx} 
\usepackage{algorithm, algorithmic}
\usepackage{color}
\usepackage{amsmath, amscd, amsthm, amsfonts, amssymb, times}
\usepackage{url}
\usepackage{multirow}
\usepackage{booktabs}
\usepackage{abstract}
\usepackage{textcomp}
\usepackage{newunicodechar}
\newunicodechar{₹}{\textrupee}
\usepackage{fancyhdr}
\usepackage{booktabs}
\newtheorem{lemma}{Lemma}
\usepackage[table,xcdraw]{xcolor}
\usepackage{subcaption}
\usepackage[table,xcdraw]{xcolor} 
\usepackage{colortbl} 
\usepackage{cite}
\usepackage{comment}
\usepackage{soul}
\definecolor{lightgreen}{RGB}{220,255,220}
\definecolor{lightred}{RGB}{255,230,230}
\definecolor{lightyellow}{RGB}{255,255,200}

\title{
    \textbf{F\&O Expiry vs. First-Day SIPs:} \\
    \textbf{A 22-Year Analysis of Timing Advantages in India's Nifty 50} 
}

\author{
    Siddharth Gavhale \\
    School of Interdisciplinary Studies \& Research \\
    D Y Patil International University \\
    Akurdi, Pune 411044, India \\
    \texttt{siddharth.gavhale@dypiu.ac.in}
}
\date{}

\begin{document}
\maketitle
\begin{abstract}
Systematic Investment Plans (SIPs) are a primary vehicle for retail equity participation in India, yet the impact of their intra-month timing remains underexplored. This study offers a 22-year (2003–2024) comparative analysis of SIP performance in the Nifty 50 index, contrasting the conventional first-trading-day (FTD-SIP) strategy with an alternative aligned to monthly Futures and Options expiry days (EXP-SIP). Using a multi-layered statistical framework—including non-parametric tests, effect size metrics, and stochastic dominance—we uncover two key findings. First, EXP-SIPs outperform FTD-SIPs by 0.5–2.5\% annually over short-to-medium-term horizons (1–5 years), with Second-Order Stochastic Dominance (SSD) confirming the EXP-SIP as the preferred choice for all risk-averse investors. Second, we establish boundary conditions for this advantage, showing it decays and becomes economically negligible over longer horizons (10–20 years), where compounding and participation dominate. Additionally, the study challenges the prevalent “12\% return” narrative in Indian equity markets, finding that the 20-year pre-tax CAGR for Nifty 50 SIPs is closer to 6.7\%. These findings carry significant implications for investor welfare, financial product design, and transparency in return reporting.
\end{abstract}



\vspace{0.5cm} \hrule \vspace{0.5cm}
\section{Introduction}
Systematic Investment Plans (SIPs) have emerged as a cornerstone of retail investment in India's equity markets, offering a structured and disciplined path to long-term wealth creation. Designed to reduce the emotional and cognitive burden of investment decisions, SIPs enable individuals to invest fixed amounts at regular intervals, thereby benefiting from rupee-cost averaging and compounding returns over time. This approach is particularly appealing in India--a rapidly growing economy with a rising middle class seeking accessible wealth-creation tools \cite{WorldBank2021, RBI2022}. Among the available investment avenues, the Nifty 50 Index has emerged as a widely adopted benchmark for systematic investment plans (SIPs) in India. The index comprises 50 large-cap companies across 13 sectors and is considered a proxy for the Indian equity market and the broader economy \cite{nifty_methodology}.
According to recent data from the Association of Mutual Funds in India (AMFI), as of February 2025, more than 100 million investors participate in SIPs. Among those, 4.456 million are new SIP accounts, further contributing approximately INR 25,999 crore monthly to mutual funds \cite{amfi_feb2024}. This marks a significant rise from early 2021, underscoring the growing popularity of the SIP route among retail investors \cite{amfi_historical}.
Despite their popularity, a critical dimension of SIPs remains underexamined: the intra-month timing of execution. A vast number of SIPs are scheduled by default for the first trading day of the month, a convention likely driven by salary credit cycles. However, this uniformity may inadvertently expose a majority of investors to systematically unfavorable market entry points if predictable microstructure effects are present.

One such predictable pattern occurs around monthly Futures and Options (F\&O) expiry days—typically the last Thursday of each month—which are marked by elevated trading volumes and short-term price dislocations. 
Prior research documents significant volume spikes and systematic price dips just before expiry, often followed by sharp rebounds driven by institutional arbitrage and liquidity constraints in the cash market \cite{Vipul2005}. This raises a compelling question: can aligning SIP execution with this predictable expiry-day volatility yield superior outcomes?
While earlier studies confirm the long-run superiority of SIPs over lump-sum strategies \cite{Venkataramani2021}, they do not explore the tactical potential of intra-month timing. This gap persists in both academic and practitioner discourse. Consequently, whether expiry-related volatility can be harnessed as a source of alpha for systematic investors remains an open empirical question.

This paper addresses this gap through a 22-year comparative analysis (2003–2024) of two SIP timing strategies using the Nifty 50 index: (i) the FTD-SIP, executed on the first trading day of the month, and (ii) the EXP-SIP, executed on the monthly F\&O expiry day. We evaluate relative performance across a range of horizons and uncover three key findings:
\begin{enumerate}
    \item \textit{A Robust Timing Advantage Exists:} The EXP-SIP strategy generates a statistically and economically significant outperformance over the FTD-SIP strategy in 1- to 5-year horizons, proving that intra-month timing is not irrelevant.
    \item \textit{The Advantage is Horizon-Dependent:} The nature of the outperformance is dynamic. We demonstrate that the EXP-SIP strategy achieves Second-Order Stochastic Dominance, making it preferable for any risk-averse investor, but the drivers of this advantage shift as the investment horizon lengthens.
    \item \textit{Long-Term Returns are Widely Misrepresented:} Our data directly challenges the widely marketed \text{12–15\% return} narrative. The actual 20-year pre-tax CAGR for the Nifty 50 is closer to 6.7\%, revealing a critical disconnect between industry claims and empirical reality.
\end{enumerate}
Notably, the tactical advantage of the EXP-SIP strategy diminishes over multi-decade horizons, reinforcing the foundational principle that the discipline of consistent, uninterrupted participation remains the dominant driver of long-term wealth.
The remainder of this paper is structured as follows: Section 2 reviews theoretical and empirical literature on SIP performance and timing. Section 3 outlines the methodology applied to 22 years of Nifty 50 index data. Section 4 presents theoretical and empirical results comparative performance results across different investment durations. Section 5 discusses broader implications and limitations of the findings, and Section 6 concludes the study.
\section{Literature review}
Systematic Investment Plans sit at the intersection of classical finance theory, behavioral economics, and market microstructure. From the classical lens, SIPs derive legitimacy from the Efficient Market Hypothesis (EMH), which posits that asset prices incorporate all publicly available information, thereby rendering consistent market timing futile \cite{Fama1970}. This perspective, popularized by Malkiel’s \textit{Random Walk }theory \cite{Malkiel1989}, positions SIPs as tools to ensure \textit{time in the market} rather than vehicles for superior timing. The underlying mechanism of SIPs—Dollar-Cost Averaging (DCA)—is thus viewed as a discipline-enforcing strategy rather than a return-enhancing one.
However, classical finance has long criticized DCA as suboptimal. Constantinides \cite{Constantinides1979} argued that DCA systematically withholds capital from the market, causing investors to forgo the equity risk premium they would earn under lump-sum investing. Similarly, Carhart \cite{Carhart1997} demonstrated that even professional active fund managers struggle to outperform passive benchmarks, reinforcing the EMH-based view that timing strategies—even systematic ones—offer limited value under rational expectations.

This deterministic view has been meaningfully challenged by behavioral finance. Kahneman and Tversky’s seminal Prospect Theory \cite{Kahneman1979} established that investor behavior often deviates from rational utility maximization due to cognitive biases such as loss aversion and framing effects. Retail investors, in particular, are known to exhibit herd behavior, recency bias, and the disposition effect—the tendency to prematurely sell winners and hold onto losers \cite{Odean1998}. Barber, Odean, and Zhu \cite{Barber2009} found that retail investor trades are systematically correlated and that the most aggressively purchased stocks tend to underperform in subsequent periods, indicating persistent timing errors. Within this behavioral paradigm, SIPs function as commitment devices—predefined investment rules that help investors sidestep their own cognitive pitfalls and emotional biases.
Yet, even this behavioral framing overlooks a deeper layer of analysis: the predictable frictions arising from market microstructure. Markets are not perfectly efficient or random; they are shaped by the institutional mechanics of trading, settlement, and liquidity flows. One particularly relevant friction is \text{flow-driven price pressure}, wherein coordinated, inelastic demand temporarily distorts asset prices. Coval and Stafford \cite{Coval2007} demonstrated that large, predictable flows from mutual funds exert temporary upward pressure on prices, leading to subsequent underperformance. Edelen, Ince, and Kadlec \cite{Edelen2016} further confirmed that such flow-induced distortions raise transaction costs and degrade long-term returns.
This friction is directly relevant to the Indian SIP landscape, where inflows are heavily concentrated on the initial trading day of each month due to payroll cycles \cite{SEBI2023}. Retail investors, acting in aggregate, create a predictable surge of demand at month-beginning, generating a temporary price inflation that long-term SIP investors unwittingly absorb. This \textit{salary-day effect} can be viewed as a structural inefficiency that systematically penalizes disciplined retail investors by forcing purchases into periods of peak demand and elevated valuations.

Juxtaposed against this retail-driven inefficiency is another calendar-based anomaly: the derivatives expiry effect. On the last Thursday of each month, the Indian equity market experiences heightened volatility and price dislocations due to the unwinding of futures and options (F\&O) positions. As options near expiration, their delta becomes increasingly sensitive to changes in the underlying asset—an effect known as high gamma. This creates strong incentives for institutional players to “pin” prices near high open interest strike levels, thereby maximizing the value of their derivatives positions \cite{Stoll1991,  Ni2005}. Vipul \cite{Vipul2005} provided early empirical evidence of such expiry-driven volatility and price clustering in the Nifty spot market, attributing it to both gamma hedging by market makers and deliberate strategic behavior by institutions.

Recent events provide a stark illustration of such institutional manipulation of market structure. In mid 2025, Securities and Exchange Board of India (SEBI) penalized the global quantitative firm Jane Street for using a 'netting-off' strategy to exploit regulatory loopholes and engineer favorable price movements \cite{JaneStreet2025}. This incident serves as a powerful real-world example of the strategic exploitation described in the academic literature. While often associated with the high-frequency trading (HFT) 'arms race' to exploit microsecond latencies \cite{Budish2015}, the underlying principle is broader: sophisticated institutions possess the capability and incentive to manipulate market mechanics for profit. The Jane Street case is not merely an HFT anomaly; it is crucial evidence that lends significant credence to our hypothesis that similar strategic behavior is employed in the cash market to influence the outcome of large derivatives positions on expiry day.
Taken together, these two effects—the retail-driven price inflation on salary day and the institutionally-driven dislocations on expiry day—create a dichotomous structure within each calendar month. The first represents coordinated, unsophisticated buying pressure; the second, strategic price manipulation by elite players. To the best of our knowledge, no empirical study has specifically evaluated whether SIP investors can enhance long-term outcomes by adjusting the intra-month timing of their investments.

This paper addresses that gap. Building on established insights from EMH, behavioral finance, and market microstructure theory, we formulate and test a novel hypothesis: that an expiry-day SIP strategy—executed on the last Thursday of each month—can outperform the first-day SIP by systematically exploiting price dislocations created by institutional trading dynamics, while avoiding retail-driven price pressure. Through a robust, long-term empirical analysis of Nifty 50 indic returns in India, we evaluate the comparative effectiveness of these two SIP strategies across multiple investment horizons.

\section{Methodology}
This study evaluates the impact of timing and duration on the performance of Systematic Investment Plans (SIPs) by simulating all returns using historical daily closing prices of the Nifty 50 index over the 2003–2024 period extracted from  NSE India \cite{NiftyHistoricalData}. This window is selected to ensure data completeness following the introduction of stock and index derivatives (Futures \& Options) in India during 2000–2001\cite{NSEHistory}. Although direct investment in the index is not possible, the Nifty 50 serves as a practical proxy for passive market exposure and is widely used as a benchmark by index funds and ETFs \cite{Agapova2011}. To ensure realistic and replicable calculations, daily closing prices are used for all return computations, as intraday highs and lows are not consistently accessible and are operationally impractical for SIP execution—even for high-frequency systems \cite{Virgilio2019}. All SIPs are simulated using calendar-year windows (January to December), maintaining uniformity across scenarios. A detailed mapping of monthly first trading days and F\&O expiry dates used for investment timing is provided in Appendix Table \ref{tab:ftd_exp}.

\subsection{Investment Durations}
To assess the influence of investment horizon on SIP performance, five distinct and non-overlapping durations are analyzed within the 2003--2024 window. These durations are structured to represent varying investor preferences, from short-term tactical strategies to long-term wealth accumulation:

\begin{itemize}
    \item \text{1-Year SIPs:} Evaluated for each calendar year from 2003 to 2024, with 12 equal monthly contributions beginning in January and concluding in December of the same year.
    
    \item \text{3-Year SIPs:} Assessed over seven triennial intervals: 2004--2006, 2007--2009, 2010--2012, 2013--2015, 2016--2018, 2019--2021, and 2022--2024.
    
    \item \text{5-Year SIPs:} Simulated over four quinquennial periods: 2005--2009, 2010--2014, 2015--2019, and 2020--2024.
    
    \item \text{10-Year SIPs:} Represented by two decadal windows: 2005--2014 and 2015--2024.
    
    \item \text{20-Year SIP:} Captures a full investment cycle spanning the 2005--2024 period.
\end{itemize}

This structured segmentation enables consistent evaluation across timeframes while ensuring that all plans conclude by 2024. It also facilitates comparative analysis of return profiles across short-, medium-, and long-term horizons, aligning with varied investor goals.

\subsection{Investment Timing}
Beyond duration, the specific timing of monthly investments is a critical determinant of SIP returns due to market volatility and price fluctuations within each month. This study compares two practical SIP execution strategies:
\begin{itemize}
    \item \text{First Trading Day SIP (FTD-SIP):} Contributions are made on the first trading day of each calendar month. This approach reflects conventional SIP scheduling, where investments are executed early in the month, potentially allowing longer compounding for each installment.

    \item \text{F\&O Expiry Day SIP (EXP-SIP):} Contributions occur on the F\&O expiry date of the previous month, typically the last Thursday (or adjusted accordingly for holidays). This timing is chosen to explore whether market adjustments or end-of-month volatility offer any advantage.
\end{itemize}

\text{Illustrative Example:}  Consider a 1-Year SIP for the year 2010. In the FTD-SIP strategy, monthly investments are made on the first trading days of January through December 2010, with the plan maturing on December 31, 2010. In contrast, the EXP-SIP version invests on the F\&O expiry dates of the preceding month---e.g., December 31, 2009 (for January 2010), January 28, 2010 (for February), and so on---while maintaining the same maturity date.
\subsection{Return Computation and Statistical Evaluation}
To maintain analytical tractability, the study assumes direct exposure to the Nifty 50 index, thereby abstracting from complexities such as fund-level tracking errors or expense ratios. Let \( m \) represent the fixed monthly SIP contribution over an investment period of \( N \) years, with \( N \in \mathbb{N} \). The following notations are used for return computation:

\begin{itemize}
    \item \( I^F_i \): Closing index value on the first trading day of the \( i^\text{th} \) month.
    \item \( I^E_i \): Closing index value on the F\&O expiry day of the \( (i-1)^\text{th} \) calendar month.
    \item \( I^k_i \): Generalized index close on the investment day, where \( k \in \{F, E\} \) corresponds to the first trading day or expiry day strategy.
    \item \( q^k_i = \frac{m}{I^k_i} \): Units purchased in month \( i \) under strategy \( k \).
\end{itemize}

The total units accumulated across the entire investment duration is given by:
\[
    Q^k = \sum_{i=1}^{12N} q^k_i
\]

The final value of the investment at the end of the horizon is:
\begin{equation}
    F^k = I^L_{12N} \cdot Q^k,
    \label{eq:final_value}
\end{equation}
where \( I^L_{12N} \) denotes the index closing value on the last trading day of December in the terminal year. Specific maturity dates are listed in Table~\ref{tab:end_date}.
\begin{table}[htbp]
\centering
\setlength{\tabcolsep}{5pt}
\renewcommand{\arraystretch}{1.0}
\begin{tabular}{|l|l|l|l|l|l|l|l|} \hline
31/12/2024 & 29/12/2023 & 30/12/2022 & 31/12/2021 & 31/12/2020 & 31/12/2019 & 31/12/2018 & 29/12/2017 \\ \hline
30/12/2016 & 31/12/2015 & 31/12/2014 & 31/12/2013 & 31/12/2012 & 30/12/2011 & 31/12/2010 & 31/12/2009 \\ \hline
31/12/2008 & 31/12/2007 & 29/12/2006 & 30/12/2005 & 31/12/2004 & 31/12/2003 &             &            \\ \hline
\end{tabular}
\caption{Last trading date of year (2003-2024)}
\label{tab:end_date}
\end{table}

To evaluate investment performance, the Compound Annual Growth Rate (CAGR) is used. CAGR reflects the annualized rate of return assuming reinvestment and compounding. While XIRR is widely used for SIP returns, its reliance on internal rate of return calculations introduces the possibility of multiple solutions under non-conventional cash flow patterns \cite{WikipediaIRR}, a challenge acknowledged in both academic and practitioner literature \cite{Damodaran2007}. Additionally, absolute return, though simple, fails to capture investment duration, making it less suitable for comparing long-term performance \cite{FundsIndiaAbsoluteReturn}. The CAGR is calculated using:
\begin{equation}
   \text{CAGR} = \left[\left( \frac{\text{Final value}}{\text{Total invested}} \right)^{\frac{1}{N}} - 1 \right] \times 100
   = \left[\left( \frac{F^k}{12mN} \right)^{\frac{1}{N}} - 1 \right] \times 100
   \label{eq:cagr}
\end{equation}
For each of the 36 investment plans across all durations and schedules (FTD and EXP), we compute two separate CAGR values—\textit{CAGR\_f} and \textit{CAGR\_e}—to compare the impact of timing on final returns. This allows us to examine whether entry-day strategy has a systematic influence on performance, particularly in the presence of market volatility and cyclicality.

\subsection{Statistical Evaluation Framework}

In this work, empirical analysis employs a multilayered statistical framework designed to comprehensively evaluate the performance differential between EXP-SIP and FTD-SIP strategies. This approach systematically addresses three fundamental questions: whether the observed outperformance is statistically significant, whether it is economically meaningful, and whether it persists across the entire return distribution for various investor types. By integrating complementary analytical techniques, we ensure robust conclusions that are neither artifacts of specific statistical assumptions nor limited to particular performance metrics.
The evaluation begins with paired hypothesis testing to establish baseline statistical significance. 
We first implement a conventional one-tailed paired t-test of the null hypothesis that the expiry-day strategy does not outperform 
($H_0: \mu_{\text{EXP}} - \mu_{\text{FTD}} \leq 0$) against our alternative hypothesis of systematic advantage ($H_1: \mu_{\text{EXP}} - \mu_{\text{FTD}} > 0$). 
However, recognizing the well-documented non-normality characteristic of financial returns \cite{Cont2001},
we complement this parametric test with the non-parametric Wilcoxon signed-rank procedure. 
This dual-test approach ensures our inferences about central tendency are robust to distributional assumptions, 
heavy tails, or outlier effects that frequently appear in financial datasets \cite{Wilcox2017}.

While hypothesis testing establishes statistical significance, it provides limited insight into the practical importance of any observed differences. Therefore, the analysis further supplemented with multiple complementary measures of economic significance. 
The standardized effect size, measured through both Cohen's $d$ and its small-sample adjusted counterpart Hedges' $g$, quantifies the mean difference in units of variability, allowing comparison across different investment horizons.
Following conventional benchmarks \cite{Cohen1988, Hedges1981}, we interpret absolute values exceeding 0.2 as economically meaningful, with values above 0.5 representing substantial effects and values above 0.8 representing large postive effect. This multi-metric assessment prevents overreliance on any single measure of practical significance while providing intuitive interpretation of the results' investment implications.
To account for potential small-sample biases and distributional irregularities, we implement a non-parametric bootstrap procedure that makes minimal assumptions about the underlying data generating process. Drawing 10,000 resamples of the paired differences, we construct bias-corrected and accelerated (BCa) confidence intervals that adjust for both skewness and median bias \cite{Efron1986}. This computationally intensive approach yields second-order accurate confidence limits with superior coverage probabilities compared to conventional methods, providing reliable inference regardless of the true distributional characteristics of SIP returns.

The most theoretically grounded component of our analysis evaluates stochastic dominance relationships between the strategies' return distributions. Using the consistent test procedure developed by \cite{Davidson2000}, we examine whether the expiry-day strategy demonstrates first-order stochastic dominance (preferred by all rational investors) or second-order stochastic dominance (preferred by all risk-averse investors). These tests move beyond comparisons of central tendency to consider the entire distribution of outcomes, while properly accounting for the statistical dependence between the paired strategies. The resulting dominance relationships provide the strongest form of comparative inference, directly connecting our empirical results to theoretical models of investor choice under uncertainty.
This integrated analytical framework offers several substantive advantages over conventional single-method approaches. The combination of parametric and non-parametric tests provides robustness against distributional misspecification, while the effect size measures bridge the gap between statistical and economic significance. The bootstrap analysis yields reliable small-sample inference, and the dominance tests ground our findings in economic theory. Together, these components form a comprehensive assessment that addresses both the statistical properties and practical implications of the observed performance differences, providing actionable insights for both retail investors and financial professionals considering alternative SIP timing strategies. All statistical analyses are implemented in MatLab using reproducible codebases, ensuring transparency and replicability of findings.


\section{Results and Analysis}
\subsection{Theoretical}
The first theoretical result of this article builds upon a well-established understanding in finance: the return of an index fund is determined by the performance of the underlying index and the duration of the investment, rather than the amount invested\cite{Bodie2014 , Damodaran2012 , Sharpe1999  }. This holds true for Systematic Investment Plans (SIPs) as well, where the timing and performance of the underlying index, rather than the invested amount, influence the Compound Annual Growth Rate (CAGR).
To understand this in detail, consider the following lemma, which formalizes this concept in the context of SIP investments.
\begin{lemma}
\textit{(CAGR Dependence on Index Trajectory and Duration)}  
The Compound Annual Growth Rate (CAGR) of a SIP investment assuming direct exposure to the Nifty 50 index is independent of the invested amount. It depends solely on the investment duration and the trajectory of the index over that period.
\end{lemma}

\begin{proof}
Put \( q^k_i = \frac{m}{I^k_i} \) in \eqref{eq:cagr} gives
\[
\text{CAGR} = \left( \left( \frac{I^{L}_{12N} \left( \sum^{12N}_{i=1} \frac{1}{I^k_i} \right)}{12N} \right)^{\frac{1}{N}} - 1 \right) \times 100
\]
As the numerator and denominator both scale linearly with \( m \), the \( m \) terms cancel, leaving CAGR dependent only on the index performance and duration, i.e, $ \text{CAGR} \propto (I^k_i, N),  \text{ and CAGR} \not\varpropto m $
\end{proof}
This lemma provides a fundamental insight about SIPs: \emph{the investment amount is irrelevant to the CAGR calculation}. Instead, the timing of investments — whether made on the first trading day (FTD) or at the expiry of options (EXP) — and the market's overall trajectory play a much more critical role. An index that shows a steady upward trend with smaller drawdowns will naturally lead to a higher CAGR \cite{MomentumEffect2020}. Therefore, selecting an optimal investment schedule and understanding long-term index behavior is more impactful than merely increasing the amount invested.

\subsection{Short- and Medium-Term Investment Horizons}
This subsection synthesizes the empirical findings from multiple time horizons—ranging from 1, 3 and 5 years investment windows—based on SIP execution strategies. The Compounded Annual Growth Rate (CAGR) serves as the return metric. The data exhibits nuanced but consistent trends.

\subsubsection{One Year Investment Returns} This subsection presents a year-by-year analysis of one-year Systematic Investment Plan (SIP) returns for the Nifty 50 index between 2003 and 2024. Each year’s investment spans from January to December, using earlier mentioend two strategies. The data exhibit considerable inter-annual variability, reflecting the inherent volatility in equity markets. Years such as 2003, 2005, 2006, 2007, 2009 \& 2020,  recorded CAGRs exceeding 20\% under both strategies. In 2003, the first-day investment yielded a CAGR of 60.43\% compared to 62.94\% from expiry-based SIPs. See Table \ref{tab:result} in Appendix and Figure \ref{fig_y} for details. Conversely, bearish years such as 2008, 2011 and 2015 experienced negative returns. The 2008 financial crisis led to nearly identical CAGRs: -30.08\% for the first-day strategy and -29.99\% for expiry-day SIPs, with an insignificant difference. In 2011, both strategies again posted losses, with expiry-day investments slightly outperforming (-12.58\% vs. -14.81\%), though this difference was statistically significant. Same has repeated in 2015 with smaller difference in CAGR.
Over the entire 22-year span, the average absolute CAGR difference between the two strategies remained modest—typically within a $\pm$2.5\% range. Positive CAGR differences in favor of expiry-based investing were observed more than 80\% of the years. In 2004, 2005, and 2022, expiry-based SIPs outperformed by more than 1.55\%. 
Years like 2003, 2011, and 2013 showed CAGR differences exceeding 2\%. Meanwhile, in stable years such as \text{2015} and \text{2019}, both strategies produced nearly identical outcomes, emphasizing the limited impact of timing in less volatile periods.

 The analysis reveals that even within a relatively short one-year investment horizon, SIPs timed around the F\&O expiry date consistently delivered marginally superior returns compared to first-day-of-month investments. This suggests that expiry-based timing is not just a random anomaly but may be capturing systematic behavioral or structural market effects. While the average annual CAGR difference remained within $\pm$2.5\%, the persistent directional advantage in favor of expiry-based SIPs highlights its potential as a tactical enhancement, especially in turbulent or momentum-driven markets. Thus, even for short-term SIPs, aligning investments with expiry dates may offer a subtle yet repeatable edge - reinforcing the case for expiry-aware SIP execution strategies.
\begin{figure}[ht!]
\centering{
\includegraphics[scale=0.15]{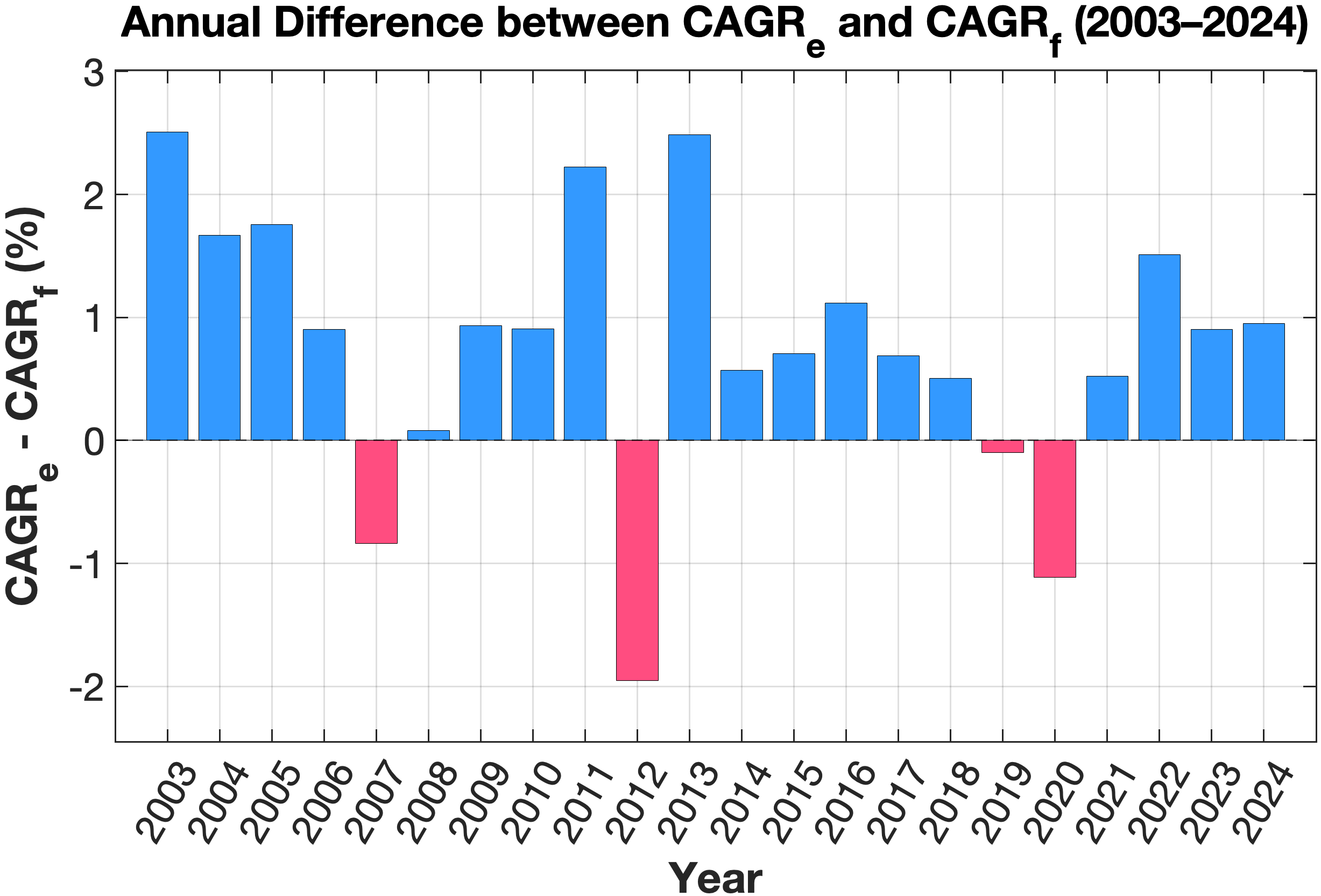}\\
\includegraphics[scale=0.078]{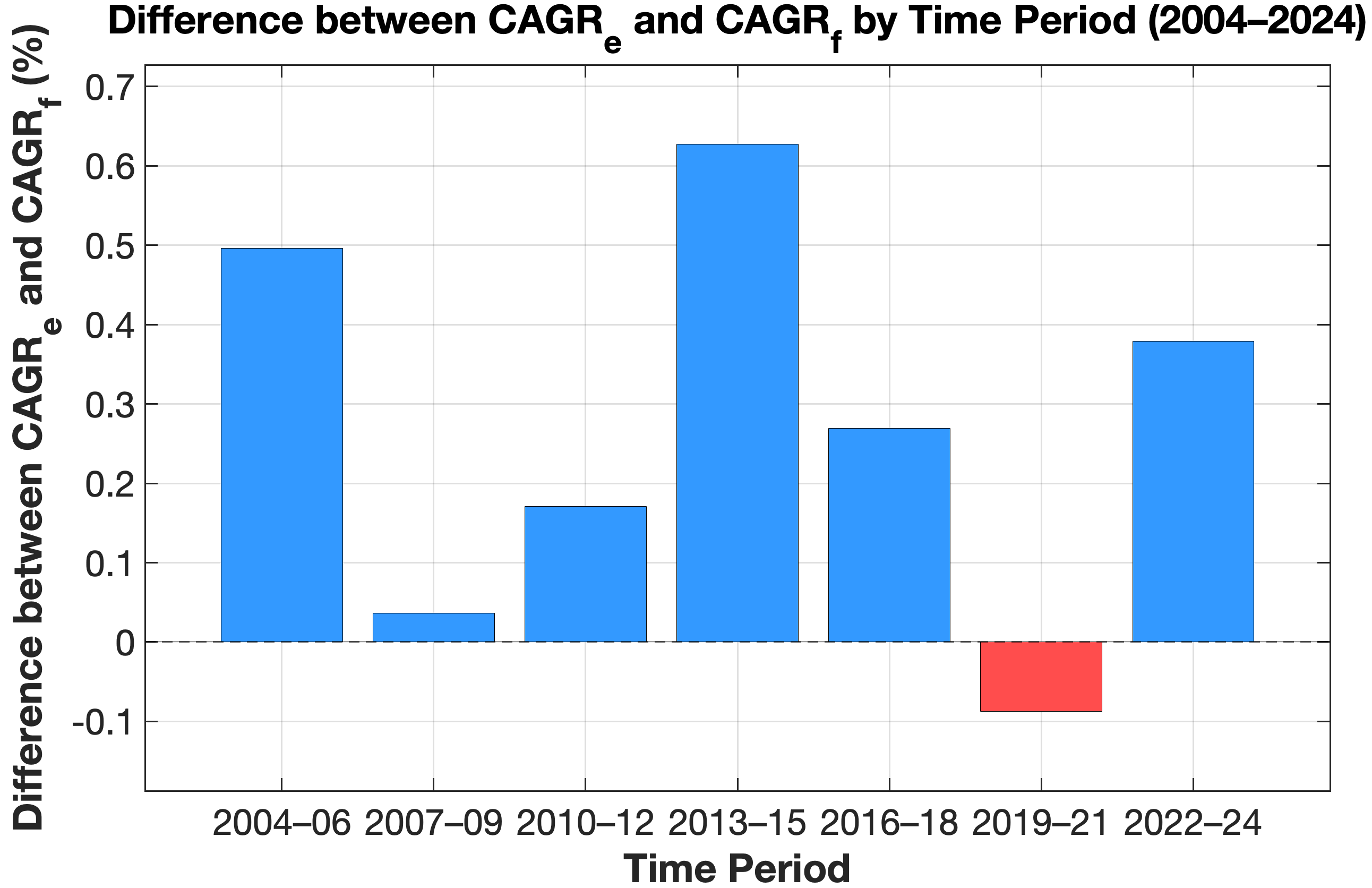}
\includegraphics[scale=0.08]{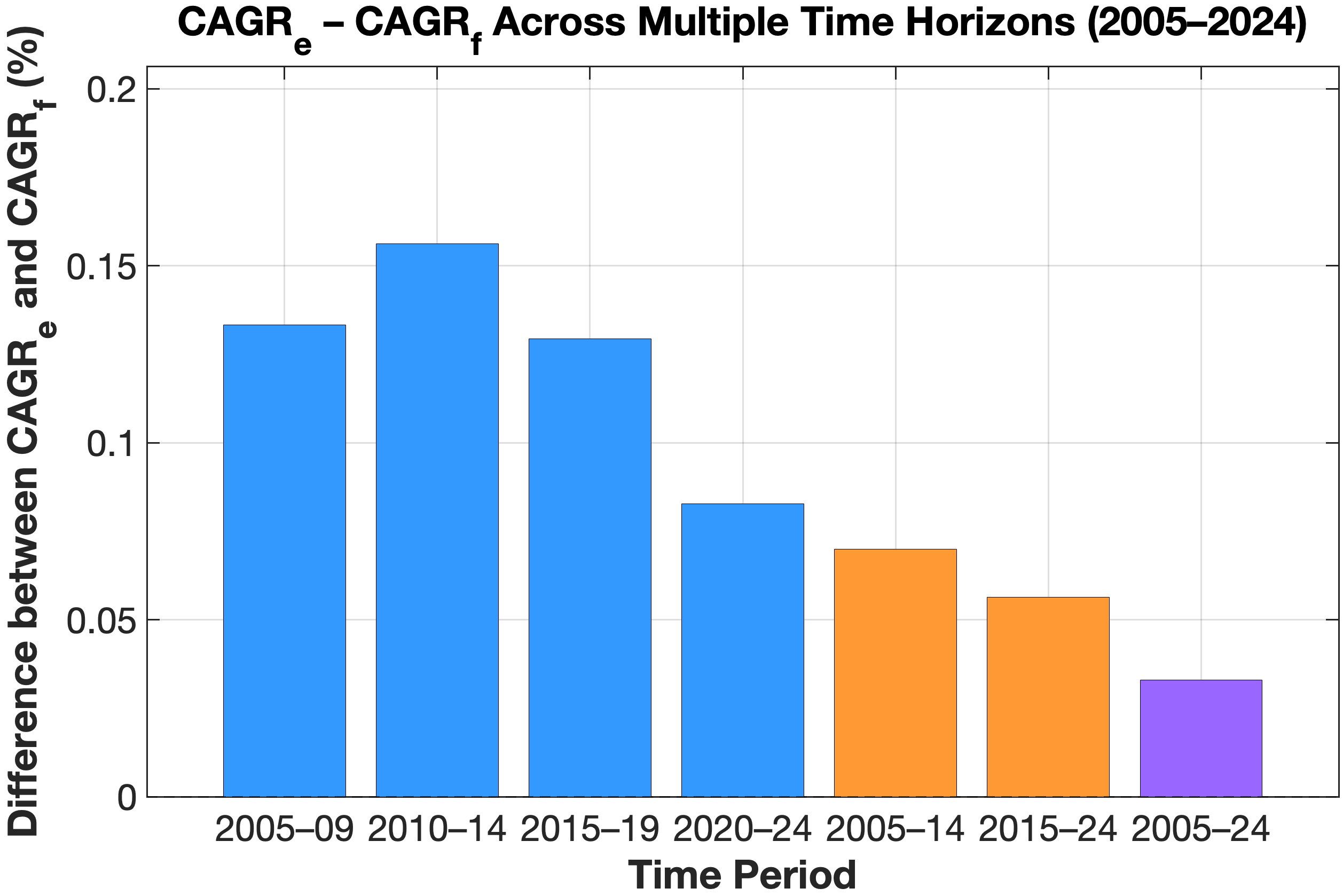}
\caption{ Year-wise difference in three-year SIP returns between EXP-SIP and FTD-SIP}
\label{fig_y} 
}
\end{figure}

 \subsubsection{Three-Year Investment Returns}
The analysis of three-year rolling SIP investments reveals a clearer picture of performance differentials between the two strategies by reducing the noise of short-term market volatility. Over the seven three-year periods from 2004–2006 to 2022–2024, both the first-day and expiry-day strategies showed reasonably strong growth, with CAGR values ranging from approximately 3\% to 21\%. Notably, the expiry-based strategy outperformed the first-day strategy in six out of seven periods, with CAGR advantages typically ranging between 0.04\% and 0.63\%, see Figure \ref{fig_y} and Table \ref{tab:result} for details. 
Importantly, unlike the more erratic one-year results, the three-year intervals display a smoother, more predictable trend, reinforcing the importance of extending investment horizons in SIP-based strategies. The observed edge in expiry-day investments during multiple periods may reflect favorable end-of-month liquidity patterns or systematic institutional behaviors around derivatives expiry that influence price levels.
The three-year investment analysis strengthens the case for expiry-based SIP timing. 
The strategy not only outperformed more frequently but also did so with statistical significance across the observed periods.
 While the absolute differences in CAGR are modest, their consistent direction and significance suggest that expiry-based timing may capture a persistent structural advantage. 
 Thus, even over medium-term horizons, aligning SIP execution with F\&O expiry may offer a subtle but repeatable edge for disciplined investors.

\subsubsection{Five Year Investment Returns}

The five-year investment horizon offers a more stable 
perspective on SIP performance by dampening
 short-term volatility and better capturing broader market cycles. 
Across the four evaluated five-year periods—2005–2009, 2010–2014, 2015–2019, and 
2020–2024—both SIP strategies generated consistent growth, with CAGR values typically ranging between 5\% and 9\%. Full results are presented in Table \ref{tab:result} and visualized in Figure \ref{fig_y}.
Importantly, the expiry-based strategy outperformed the first-day strategy in all four periods, with positive CAGR differentials ranging from approximately 0.08 to 0.16 percentage points. This consistent directional advantage lends credibility to the hypothesis that expiry-aligned SIPs benefit from favorable structural or behavioral market effects around the derivatives expiry. The earliest period (2005–2009), covering the post- Global Financial Crisis (GFC) rebound, showed the largest outperformance (9.04\% vs.\ 8.91\%).
Though the return gaps in the subsequent periods were more modest, the fact that all differences remained consistently positive suggests a subtle structural edge rather than random noise.  The five-year results robustly support the notion that expiry-day SIPs may provide a modest but persistent performance edge over traditional first-day investments. While the absolute advantage is small, its consistency—especially during volatile recovery phases—strengthens the argument for expiry-aware SIP execution. This finding affirms that even over longer horizons, intelligent timing can act as a complementary force to disciplined investing.
\begin{figure}[ht!]
\centering{
\includegraphics[scale=0.65]{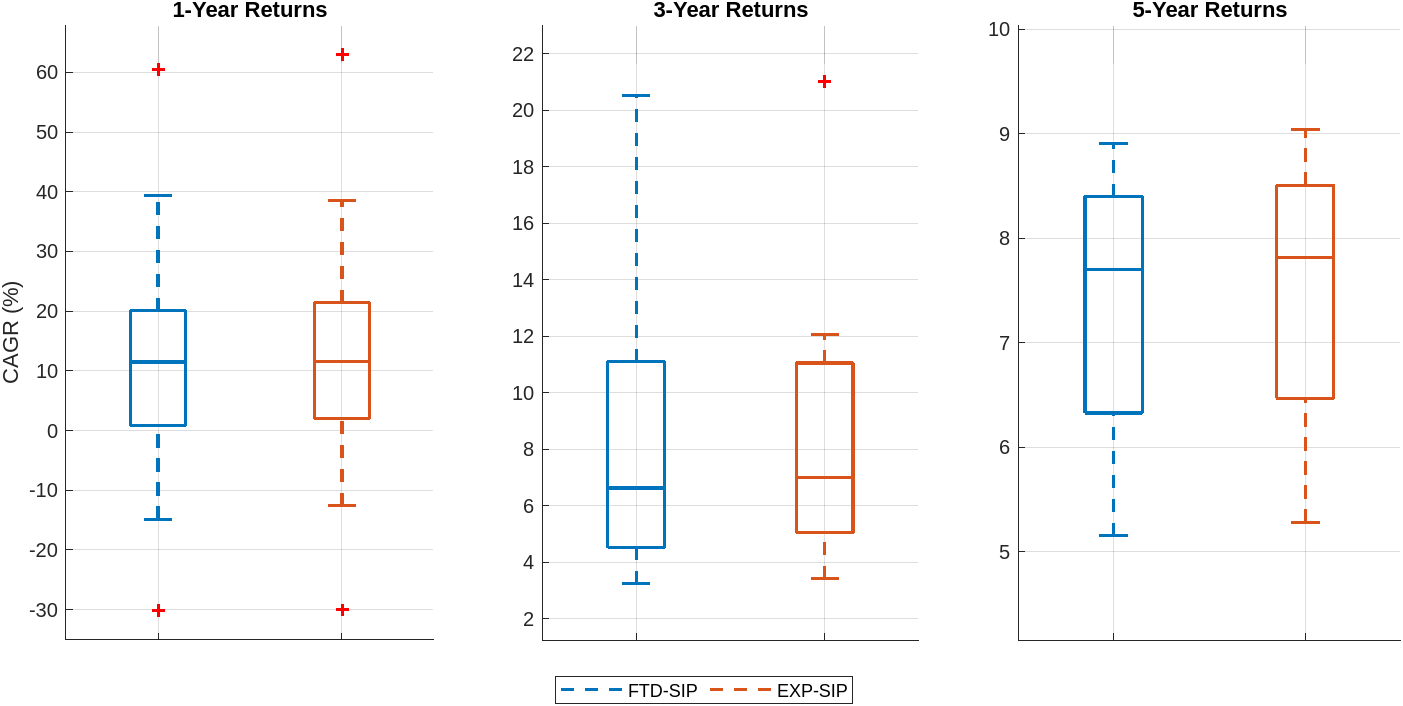}
\caption{ Distribution of SIP returns across investment horizons. Boxplots compare first trading day (FTD) and expiry-day (EXP) SIP performance for (Left) 1-year, (Middle) 3-year, and (Right) 5-year investment periods in Nifty 50 (2003–2024). Boxes show quartiles (25th–75th percentile), and outliers (+) are values beyond the whiskers. Blue: FTD-SIP, Orange: EXP-SIP.}
\label{fig_box} 
}
\end{figure}

\subsubsection{Distributional Analysis}
An analysis of the return distributions, presented in Figure~\ref{fig_box}, reveals a modest but persistent, multi-dimensional advantage for the expiry-day (EXP-SIP) strategy.
EXP-SIPs consistently deliver higher median returns than first-day SIPs (FTD-SIPs) across all evaluated horizons. However, this advantage is non-monotonic—peaking at 38 basis points for the 3-year horizon, and narrowing to 3 and 12 basis points for the 1- and 5-year periods, respectively. This pattern suggests that the timing edge linked to central tendency is most pronounced in the medium term.
A more consistent advantage lies in the EXP-SIP strategy's superior downside protection. The 25th percentile of its returns remains above that of the FTD-SIP across all horizons, with the largest margin of 111 basis points occurring in the volatile 1-year horizon. This implies improved resilience in adverse market phases, a key concern for risk-averse investors.
Regarding overall dispersion, the salient trend is the dramatic contraction of the interquartile range (IQR) for both strategies with increasing investment horizon, reflecting the volatility-smoothing effect of long-term investing. While EXP-SIPs exhibit marginally tighter dispersion in the 3- and 5-year periods, return variability is broadly similar between the two strategies.

In synthesis, although return distributions significantly overlap, the EXP-SIP demonstrates a structurally enhanced profile—offering non-monotonic but persistent median outperformance and consistently better downside protection. This improved risk-return characteristic provides a compelling rationale for adoption, plausibly linked to expiry-related liquidity and volatility dynamics in the Indian equity market.

\subsection{Statistical Validation of the Expiry-Day Advantage}

To formally assess the superiority of the EXP-SIP strategy over the FTD-SIP, we conduct a multilayered statistical analysis encompassing central tendency, effect sizes, distributional robustness, and utility-theoretic dominance. The consolidated results in Table \ref{tab:sip_results_summary} are analyzed thematically below, providing convergent evidence from multiple statistical perspectives.

\begin{table}[h!]
\centering
\rowcolors{2}{white}{gray!5}
\begin{tabular}{lccc}
\toprule
\rowcolor{gray!20}
\textbf{Metric} & \textbf{1-Year} & \textbf{3-Year} & \textbf{5-Year} \\
\midrule
\textbf{Sample Size (n)} & 22 & 7 & 4 \\
\rowcolor{lightgreen}
\textbf{Paired t-test (p-value)} & 0.0018 & 0.0149 & 0.0019 \\
\textbf{t-statistic} & 3.271 & 2.833 & 8.138 \\
\rowcolor{lightgreen}
\textbf{Wilcoxon Signed-Rank (p-value)} & 0.0035 & 0.0234 & 0.0625 \\
\textbf{Effect Size (Cohen's d)} & 0.697 (medium) & 1.071 (large) & 4.069 (very large) \\
\textbf{Effect Size (Hedges' g)} & 0.671 & 0.902 & 2.325 \\
\rowcolor{lightgreen}
\textbf{Bootstrap 95\% CI (Mean Diff)} & [0.301, 1.203] & [0.099, 0.443] & [0.095, 0.150] \\
\textbf{Bootstrap Point Estimate (E–F)} & 0.770 & 0.270 & 0.125 \\
\textbf{Stochastic Dominance: FSD} &  (KS p=0.6391) & (KS p=0.8424) & (KS p=0.7227) \\
\rowcolor{lightgreen}
\textbf{Stochastic Dominance (SSD)} & EXP-SIP SSD &  EXP-SIP SSD &  EXP-SIP SSD \\
\bottomrule
\end{tabular}
\caption{Comparison of SIP Performance Metrics across Investment Horizons}
\label{tab:sip_results_summary}
\end{table}

\subsubsection{Tests of Central Tendency}
The superiority of the EXP-SIP strategy is confirmed through both parametric and non-parametric tests. The paired-sample $t$-test reveals that the mean returns are significantly higher across all investment horizons: 1-year ($p=0.0018$), 3-year ($p=0.0149$), and 5-year ($p=0.0019$). This finding is reinforced by the non-parametric Wilcoxon signed-rank test for the 1- and 3-year horizons. While the 5-year Wilcoxon test falls just short of conventional significance ($p=0.0625$), this is attributable to the low statistical power of the test with an extremely small sample ($n=4$). The strong and significant $t$-statistic for this horizon ($t=8.138$) suggests the mean advantage remains robust.

\subsubsection{Effect Size, Economic Significance, and Robustness via Resampling}
The economic magnitude of the EXP-SIP advantage increases with investment horizon. The effect size, measured by Cohen's $d$, escalates from "medium" ($d=0.697$) at the 1-year horizon to "very large" ($d=4.069$) at the 5-year horizon, indicating a substantial real-world advantage that grows over time. This upward trend highlights how small, consistent monthly advantages in EXP-SIP returns compound over time into economically material gains—especially relevant for retail investors focused on long-term wealth accumulation.
The reliability of this effect is validated through non-parametric bootstrap resampling. The 95\% bias-corrected and accelerated (BCa) confidence intervals for the mean difference lie strictly above zero for all horizons, providing strong, assumption-free evidence of a positive return advantage. Notably, the intervals narrow with longer horizons (e.g., from [0.301, 1.203] at 1-year to [0.095, 0.150] at 5-year), reflecting increased statistical precision and reduced return volatility over time.  These results confirm that the observed advantage is not a sample artifact but holds under minimal assumptions.

\subsubsection{Utility-Theoretic Dominance}

From a decision-theoretic standpoint, the EXP-SIP strategy demonstrates a compelling advantage grounded in investor welfare. While first-order stochastic dominance (FSD) is not observed—primarily due to overlapping returns during brief periods of market distress—we establish conclusive \text{second-order stochastic dominance (SSD)} across all three investment horizons. This result is critically important, as it implies that any rational, risk-averse investor with a concave utility function would derive higher expected utility from the EXP-SIP return distribution.
Importantly, the strength of this dominance intensifies with the investment duration. At the 1-year horizon, the EXP-SIP strategy already satisfies SSD despite greater short-term volatility. By the 3-year mark, contracting return dispersion enhances the clarity of the dominance, and by 5 years, the distributions are so cleanly separated (as visualized in Figure~\ref{fig_box}) that the SSD condition holds across virtually the entire support. Such robust and strengthening dominance is rarely observed in empirical strategy comparisons. This finding provides compelling evidence that the EXP-SIP advantage is not only statistically significant but also economically meaningful within a utility-maximizing framework.
\subsection{Long-Term Horizons: The Boundary Conditions of the Timing Advantage}

Our analysis reveals a critical, horizon-dependent decay in the efficacy of the EXP-SIP strategy. While demonstrating robust outperformance in 1- to 5-year windows, this advantage systematically attenuates over longer timeframes.
For the two 10-year windows, the performance with EXP-SIP yields $6.83\% vs. 6.76\%$ and $7.04\% vs. 6.99\%$. On the other hand, 20-year horizon EXP-SIP yields $6.68\% vs. 6.65\%$ return.
With performance differences narrowing to economically negligible margins of 7, 5, and 3 basis points for the two 10-year and single 20-year periods in our dataset, see Table \ref{tab:result_10_20}. Though the limited sample size precludes formal statistical inference, the directional consistency of these results provides strong descriptive evidence of convergence.
\begin{table}[ht!]
\centering
\begin{tabular}{|c|c|c|r|r|>{\columncolor[HTML]{D9EAD3}}r|}
\hline
\rowcolor[HTML]{D9EAD3} 
\textbf{From Year} & \textbf{To Year} & \textbf{Years} & \textbf{CAGR (F)} & \textbf{CAGR (E)} & \textbf{Difference} \\
\hline
2005 & 2014 & 10 & 6.76 & 6.83 & \cellcolor[HTML]{D9F7BE}0.07 \\
2015 & 2024 & 10 & 6.99 & 7.04 & \cellcolor[HTML]{D9F7BE}0.05 \\
\hline
\hline
2005 & 2024 & 20 & 6.65 & 6.68 & \cellcolor[HTML]{D9F7BE}0.03\\
\hline
\end{tabular}
\caption{Long term SIP Investment Performance: First Day vs F\&O Expiry}
\label{tab:result_10_20}
\end{table}

This empirical pattern aligns precisely with fundamental financial principles. First, the \textit{Law of Large Numbers} suggests that the impact of monthly entry-point variations becomes asymptotically irrelevant as the number of investment periods grows. Second, the absorption of multiple complete market cycles (including the Global Financial Crisis and COVID-19 pandemic in our sample) creates an endogenous smoothing effect that neutralizes timing-specific advantages. Finally, any microstructure effects underlying the short-term advantage become mathematically insignificant when amortized across 120 to 240 monthly intervals.
This evidence establishes clear boundary conditions for the timing strategy. While economically meaningful for horizons under five years, the effect size systematically decays with investment duration, becoming practically irrelevant beyond a decade. The analysis highlights a fundamental shift from timing-sensitive to discipline-dependent wealth accumulation as investment horizons extend.
\section{Discussion}
Our empirical analysis provides two primary contributions with significant implications for both investment theory and practice. First, we document a robust, horizon-dependent timing advantage for expiry-day SIPs, challenging the notion that intra-month timing is irrelevant. Second, our long-term return analysis serves as a crucial reality check against the widely publicized but empirically unsupported myth of 12-15\% annualized returns from index fund SIPs in the Indian context.

\subsection{The Duality of the Timing Advantage: Interpretation and Evidence}
The core finding of this paper is the dual nature of the SIP timing effect: a potent advantage in the short-to-medium term that decays into irrelevance over multi-decade horizons. The robustness of the short-term advantage is validated by our comprehensive statistical framework. The superiority of the expiry-day strategy was confirmed with statistically significant paired $t$-test results across 1, 3, and 5-year horizons ($p < 0.05$ in all cases). This finding was corroborated by non-parametric Wilcoxon tests and assumption-free bootstrap analysis, which showed confidence intervals for the mean difference lying strictly above zero.
The economic magnitude of this effect, escalating to a Cohen's $d$ of over 4.0 for the 5-year period, underscores its practical relevance. Most critically, the finding of \textit{Second-Order Stochastic Dominance} (SSD) across all three horizons provides powerful, utility-theoretic proof that any rational, risk-averse investor would prefer the expiry-day strategy's return profile. We attribute this advantage to predictable market microstructure effects around derivatives expiry, such as institutional hedging, month-end rebalancing by funds, and systematic liquidity flows that create temporary price pressures.

However, as the investment horizon extends to 10 and 20 years, these small, monthly advantages are overwhelmed by the mathematical force of compounding and the averaging effect described by the Law of Large Numbers. The long-term convergence of returns underscores a foundational truth: while tactical timing can offer a discernible edge, it is the discipline of sustained participation that dominates long-term wealth creation.

\subsection{Deconstructing the "12\% Return" Myth: Biases and Real-World Frictions}
A central contribution of this study is to empirically challenge the 12-15\% annualized return benchmark frequently cited for index fund SIPs in India \cite{Financediary2024, StockInsideOut2024, ETMarkets2024}. Our data, spanning nearly two decades of market cycles from 2005 to 2024, reveals an actual pre-tax CAGR between 6.65\% and 6.68\%. This significant gap between perception and reality stems from three key distortions:

\begin{enumerate}
    \item \textit{Survivorship and Selective Period Bias:} Public narratives often anchor on exceptionally high-growth periods (e.g., 2003–2007) while downplaying subsequent stagnation or crises. This selective reporting parallels the survivorship and hindsight bias discussed in the literature on mutual fund performance illusion \cite{Brown1992 ,Wermers2000 }. When the full cycle is considered, the compounding effect of subsequent drawdowns sharply dilutes headline figures.

    \item \textit{Compounded Fee Drag:} Our frictionless estimates do not account for expense ratios (typically 1.1\% to 1.8\% annually for equity funds in India). Over decades, these fees create a substantial drag on compounded returns, a factor often overlooked in simplified projections.
    
    \item \textit{The Investor Behavior Gap:} As documented by \cite{Dichev2011}, investors often underperform the funds they invest in due to behavioral errors like panic selling or performance chasing. This "return gap" further reduces realized wealth.
\end{enumerate}

Factoring in these realities, a more realistic net, behaviorally-adjusted return expectation for a typical equity SIP would be closer to 5–6\%. This has profound implications for financial planning, retirement projections, and the fiduciary responsibility of investment advisors.

\subsection{Implications for Investor Welfare and Policy}
The inefficiencies identified—both the suboptimal default timing of SIPs and the inflated return expectations—have a substantial aggregate impact on investor welfare in India. 
With approximately 40 million active SIP accounts \cite{amfi_feb2024}, even marginal performance differences result in a cumulative wealth erosion of thousands of crores annually.
Rectifying this requires a shift in both industry practice and regulatory design. We propose three concrete interventions:

\begin{enumerate}
    \item Mandating that all mutual fund platforms and distributors offer flexible SIP timing options, including expiry-day dates, during investor onboarding.
    \item Requiring monthly expiry-related fund flow disclosures by AMCs to increase market transparency, analogous to existing FPI flow reports.
    \item Publishing official, standardized investor dashboards that report realistic SIP returns adjusted for fees, taxes, and inflation to better manage investor expectations.
\end{enumerate}

\subsection{Limitations and Future Research}
While this study offers key insights, it is subject to limitations that open avenues for future work. Our analysis is confined to the Nifty 50 large-cap index; it remains unclear if these timing effects generalize across different market caps (mid/small-cap) or sectors. Second, the findings are specific to India's market structure, and their applicability to developed markets with different derivatives cycles warrants investigation. Finally, future research could explore optimal contribution patterns (e.g., fixed vs. progressively increasing amounts) to further enhance long-term outcomes. Addressing these questions will require richer datasets and interdisciplinary approaches to inform more personalized and effective investment strategies.
\section{Conclusion}
This paper provides a comprehensive analysis of systematic investment timing in the Indian market, offering two primary contributions with significant implications for investors, advisors, and regulators. First, our research moves beyond the simplistic \textit{timing doesn't matter} narrative by empirically documenting a statistically and economically significant advantage for SIPs aligned with derivatives expiry days over short-to-medium-term horizons. Crucially, we establish clear \textit{boundary conditions} for this effect, demonstrating its systematic decay to irrelevance over multi-decade periods where the mathematical force of compounding and the discipline of continuous participation become the dominant drivers of returns.

Our second major contribution is an empirical correction to the widely marketed "12\% return myth." We demonstrate that realized, long-term, pre-tax index returns are closer to 7\%, and that this figure is further eroded by real-world frictions such as fees, biases, and investor behavior gaps. This finding exposes a critical disconnect between industry marketing and achievable outcomes, highlighting an urgent need for greater transparency.

The central implication of this work is a call for a data-driven evolution in investment practice. For the industry, it supports the design of smarter, timing-aware products and more realistic investor communication. For investors, the message is nuanced but clear: while a simple timing switch can provide a marginal, cost-free enhancement for shorter-term goals, the most critical determinant of long-term wealth creation is not the choice of day, but the unwavering discipline of consistent participation. Ultimately, while market microstructure may offer a fleeting edge, it is the steadfast commitment to staying invested that remains the most powerful and reliable strategy for building wealth.

\appendix
\section*{Appendix A: Monthly SIP Dates and Multi-Year CAGR Comparisons}
\begin{table}[ht!]
\centering
\begin{tabular}{|c|c|c|c|c|c|c|c|c|c|c|c|c|}
\hline
\textbf{Year} & \textbf{Jan} & \textbf{Feb} & \textbf{Mar} & \textbf{Apr} & \textbf{May} & \textbf{Jun} & \textbf{Jul} & \textbf{Aug} & \textbf{Sep} & \textbf{Oct} & \textbf{Nov} & \textbf{Dec} \\
\hline
2002 & & & & & & & & & & & & (26)\\
2003 & 1 (30) & 3 (27) & 3 (27) & 1 (24) & 2 (27) & 2 (24) & 1 (31) & 1 (28) & 1 (30) & 1 (30) & 3 (27) & 1 (29) \\
2004 & 1 (30) & 3 (25) & 1 (29) & 1 (30) & 3 (27) & 1 (30) & 1 (31) & 2 (29) & 1 (30) & 1 (29) & 1 (29) & 1 (31) \\
2005 & 1 (31) & 3 (28) & 1 (31) & 1 (29) & 2 (30) & 1 (29) & 1 (31) & 1 (31) & 1 (30) & 3 (28) & 1 (30) & 1 (31) \\
2006 & 1 (31) & 3 (27) & 1 (31) & 3 (27) & 2 (31) & 1 (30) & 3 (30) & 1 (31) & 1 (29) & 3 (27) & 1 (30) & 1 (31) \\
2007 & 1 (31) & 2 (26) & 3 (28) & 2 (30) & 3 (29) & 1 (30) & 2 (30) & 1 (31) & 3 (28) & 1 (31) & 1 (30) & 3 (29) \\
2008 & 1 (31) & 3 (28) & 1 (31) & 1 (30) & 2 (30) & 2 (30) & 1 (31) & 1 (31) & 1 (30) & 1 (31) & 3 (27) & 1 (31) \\
2009 & 1 (31) & 1 (28) & 3 (28) & 1 (30) & 3 (30) & 1 (30) & 1 (31) & 1 (31) & 1 (30) & 3 (28) & 1 (30) & 1 (31) \\
2010 & 3 (28) & 1 (28) & 3 (27) & 1 (30) & 1 (29) & 1 (30) & 1 (31) & 1 (31) & 2 (30) & 1 (30) & 1 (30) & 1 (31) \\
2011 & 1 (31) & 1 (28) & 2 (29) & 3 (30) & 3 (30) & 1 (30) & 1 (31) & 1 (31) & 1 (30) & 1 (30) & 2 (29) & 1 (31) \\
2012 & 1 (31) & 1 (28) & 2 (29) & 1 (30) & 1 (30) & 3 (30) & 1 (31) & 1 (31) & 1 (30) & 2 (29) & 1 (30) & 1 (31) \\
2013 & 1 (31) & 1 (28) & 1 (31) & 1 (30) & 3 (30) & 1 (30) & 1 (31) & 2 (30) & 1 (31) & 1 (30) & 1 (30) & 1 (31) \\
2014 & 1 (31) & 3 (28) & 1 (31) & 1 (30) & 2 (30) & 1 (30) & 1 (31) & 1 (30) & 3 (29) & 1 (30) & 1 (29) & 1 (31) \\
2015 & 3 (30) & 1 (29) & 1 (31) & 3 (27) & 1 (30) & 1 (29) & 1 (30) & 2 (30) & 1 (31) & 1 (30) & 3 (29) & 1 (31) \\
2016 & 1 (31) & 1 (28) & 2 (30) & 1 (30) & 1 (29) & 3 (30) & 1 (30) & 1 (31) & 3 (29) & 1 (30) & 1 (30) & 1 (31) \\
2017 & 1 (31) & 2 (29) & 3 (30) & 1 (30) & 3 (29) & 1 (30) & 1 (31) & 2 (30) & 1 (30) & 1 (31) & 1 (30) & 1 (31) \\
2018 & 1 (31) & 1 (28) & 1 (31) & 1 (30) & 2 (30) & 3 (30) & 1 (31) & 1 (31) & 1 (30) & 2 (30) & 1 (30) & 1 (31) \\
2019 & 3 (30) & 1 (28) & 2 (29) & 1 (30) & 1 (31) & 2 (30) & 1 (31) & 1 (30) & 1 (30) & 1 (30) & 1 (31) & 1 (30) \\
2020 & 1 (31) & 3 (28) & 1 (31) & 1 (30) & 2 (30) & 1 (30) & 1 (31) & 1 (31) & 3 (29) & 1 (30) & 1 (30) & 1 (31) \\
2021 & 1 (31) & 1 (29) & 3 (27) & 1 (30) & 2 (30) & 1 (30) & 1 (31) & 2 (30) & 1 (31) & 1 (31) & 1 (30) & 1 (31) \\
2022 & 1 (31) & 2 (29) & 3 (28) & 1 (30) & 2 (30) & 1 (30) & 1 (31) & 1 (31) & 2 (30) & 1 (31) & 1 (30) & 1 (31) \\
2023 & 3 (30) & 1 (28) & 2 (30) & 1 (31) & 1 (30) & 1 (31) & 2 (30) & 1 (30) & 1 (31) & 1 (30) & 1 (29) & 1 (31) \\
2024 & 1 (31) & 1 (28) & 1 (31) & 1 (30) & 2 (30) & 3 (30) & 1 (31) & 1 (31) & 2 (30) & 1 (31) & 1 (30) & 2 (\, \, ) \\
\hline
\end{tabular}
\caption{First trading date and F\&O expiry date (in parentheses) for each month}
\label{tab:ftd_exp}
\end{table}
\newpage

\begin{table}[ht!]
\centering
\begin{tabular}{|c|c|c|r|r|>{\columncolor[HTML]{D9EAD3}}r|}
\hline
\rowcolor[HTML]{D9EAD3} 
\textbf{From Year} & \textbf{To Year} & \textbf{Years} & \textbf{CAGR (F)} & \textbf{CAGR (E)} & \textbf{Difference} \\
\hline
2003 & 2003 & 1 & 60.43 & 62.94 & \cellcolor[HTML]{D9F7BE}2.51  \\
2004 & 2004 & 1 & 19.77 & 21.44 & \cellcolor[HTML]{D9F7BE}1.67  \\
2005 & 2005 & 1 & 27.40 & 29.16 & \cellcolor[HTML]{D9F7BE}1.76  \\
2006 & 2006 & 1 & 20.09 & 20.99 & \cellcolor[HTML]{D9F7BE}0.90  \\
2007 & 2007 & 1 & 39.40 & 38.56 & \cellcolor[HTML]{F7D9D9}-0.84 \\
2008 & 2008 & 1 & -30.08 & -29.99 & \cellcolor[HTML]{D9F7BE}0.09 \\
2009 & 2009 & 1 & 36.69 & 37.62 & \cellcolor[HTML]{D9F7BE}0.93 \\
2010 & 2010 & 1 & 13.85 & 14.76 & \cellcolor[HTML]{D9F7BE}0.91 \\
2011 & 2011 & 1 & -14.81 & -12.58 & \cellcolor[HTML]{D9F7BE}2.23 \\
2012 & 2012 & 1 & 11.80 & 9.85 & \cellcolor[HTML]{F7D9D9}-1.95  \\
2013 & 2013 & 1 & 6.99 & 9.47 & \cellcolor[HTML]{D9F7BE}2.48 \\
2014 & 2014 & 1 & 15.22 & 15.79 & \cellcolor[HTML]{D9F7BE}0.57  \\
2015 & 2015 & 1 & -4.60 & -3.89 & \cellcolor[HTML]{D9F7BE}0.71  \\
2016 & 2016 & 1 & 0.85 & 1.97 & \cellcolor[HTML]{D9F7BE}1.12 \\
2017 & 2017 & 1 & 11.21 & 11.90 & \cellcolor[HTML]{D9F7BE}0.69 \\
2018 & 2018 & 1 & 0.86 & 1.37 & \cellcolor[HTML]{D9F7BE}0.51  \\
2019 & 2019 & 1 & 6.70 & 6.60 & \cellcolor[HTML]{F7D9D9}-0.10  \\
2020 & 2020 & 1 & 29.62 & 28.51 & \cellcolor[HTML]{F7D9D9}-1.11 \\
2021 & 2021 & 1 & 10.64 & 11.16 & \cellcolor[HTML]{D9F7BE}0.52 \\
2022 & 2022 & 1 & 4.88 & 6.39 & \cellcolor[HTML]{D9F7BE}1.51  \\
2023 & 2023 & 1 & 16.39 & 17.29 & \cellcolor[HTML]{D9F7BE}0.90 \\
2024 & 2024 & 1 & 0.59 & 1.54 & \cellcolor[HTML]{D9F7BE}0.95 \\
\hline
\hline
2004 & 2006 & 3 & 20.53 & 21.02 & \cellcolor[HTML]{D9F7BE}0.49 \\
2007 & 2009 & 3 & 8.00 & 8.04 & \cellcolor[HTML]{D9F7BE}0.04 \\
2010 & 2012 & 3 & 3.25 & 3.42 & \cellcolor[HTML]{D9F7BE}0.17 \\
2013 & 2015 & 3 & 4.29 & 4.92 & \cellcolor[HTML]{D9F7BE}0.63\\
2016 & 2018 & 3 & 5.17 & 5.44 & \cellcolor[HTML]{D9F7BE}0.27 \\
2019 & 2021 & 3 & 12.15 & 12.06 & \cellcolor[HTML]{F7D9D9}-0.09  \\
2022 & 2024 & 3 & 6.63 & 7.01 & \cellcolor[HTML]{D9F7BE}0.38  \\
\hline
\hline
2005 & 2009 & 5 & 8.91 & 9.04 & \cellcolor[HTML]{D9F7BE}0.13  \\
2010 & 2014 & 5 & 7.50 & 7.66 & \cellcolor[HTML]{D9F7BE}0.16 \\
2015 & 2019 & 5 & 5.15 & 5.28 & \cellcolor[HTML]{D9F7BE}0.13  \\
2020 & 2024 & 5 & 7.89 & 7.98 & \cellcolor[HTML]{D9F7BE}0.09  \\
\hline
\end{tabular}
\caption{Multi-Year SIP Investment Performance: First Day vs F\&O Expiry}
\label{tab:result}
\end{table}

\newpage
\section*{Appendix B: MATLAB Code Snippet for Statistical Validation}
The following MATLAB code summarizes the analysis pipeline used:
\begin{verbatim}
% Example: Paired t-test and Cohen's d
[h, p, ~, stats] = ttest(CAGR_E, CAGR_F);
cohens_d = mean(CAGR_E - CAGR_F) / std(CAGR_E - CAGR_F);
% Wilcoxon Signed-Rank
[p_w, ~, stats_w] = signrank(CAGR_E, CAGR_F);
% Stochastic Dominance
[fE, xE] = ecdf(CAGR_E); [fF, xF] = ecdf(CAGR_F);
\end{verbatim}
All data and MATLAB code used for statistical analysis are available from the corresponding author upon reasonable request.

\section*{Acknowledgment}
The author sincerely thanks Dr.~Laxmi Priya Sahu (PhD, NIT Patna) for her valuable inputs during the drafting of this article. The author is also grateful to Ms.~Sushmita Kumari for her assistance in managing and organizing the dataset used in this study.

Special appreciation is extended to the organizers and participants of the \textit{7th International Conference on Financial Markets and Corporate Finance (ICFMCF 2025)} held at the Vinod Gupta School of Management, IIT Kharagpur. The conference offered critical feedback and expert discussion, which directly helped improve the depth and methodological rigor of this manuscript.
\section*{Declaration}
The author used ChatGPT solely for improving clarity, fluency, and grammar. All conceptual content, analysis, and written material were entirely developed and authored by the researcher.

\end{document}